\documentclass[pra,onecolumn,groupedaddress,superscriptaddress,amsmath,amssymb,a4paper]{revtex4}
\usepackage{geometry}
\usepackage{graphicx}
\usepackage{mathrsfs}
\usepackage{amssymb}
\usepackage{amsmath}
\usepackage{amsthm}
\usepackage{amsbsy}
\usepackage{bm}
\usepackage{hyperref}
\usepackage[T1]{fontenc}
\usepackage{color}
\newtheorem{proposition}{Proposition}

\newtheorem{lemma}{Lemma}

\theoremstyle{definition}

\newtheorem{example}{Example}
\newtheorem{remark}{Remark}

\newcommand{\bra}[1]{\langle #1|}
\newcommand{\ket}[1]{| #1 \rangle }

\begin{document}

\title{Phase covariant qubit dynamics and divisibility}

\author{S. N. Filippov}

\affiliation{Steklov Mathematical Institute of Russian Academy of
Sciences, Gubkina St. 8, Moscow 119991, Russia}
\affiliation{Moscow Institute of Physics and Technology,
Institutskii Per. 9, Dolgoprudny 141700, Russia}
\affiliation{Valiev Institute of Physics and Technology of Russian
Academy of Sciences, Nakhimovskii Pr. 34, Moscow 117218, Russia}

\author{A. N. Glinov}
\affiliation{Steklov Mathematical Institute of Russian Academy of
Sciences, Gubkina St. 8, Moscow 119991, Russia}
\affiliation{Moscow Institute of Physics and Technology,
Institutskii Per. 9, Dolgoprudny 141700, Russia}

\author{L. Lepp\"{a}j\"{a}rvi}
\affiliation{QTF Centre of Excellence, Department of Physics and
Astronomy, University of Turku, Turku 20014, Finland}

\begin{abstract}
Phase covariant qubit dynamics describes an evolution of a
two-level system under simultaneous action of pure dephasing,
energy dissipation, and energy gain with time-dependent rates
$\gamma_z(t)$, $\gamma_-(t)$, and $\gamma_+(t)$, respectively.
Non-negative rates correspond to completely positive divisible
dynamics, which can still exhibit such peculiarities as
non-monotonicity of populations for \textit{any} initial state. We
find a set of quantum channels attainable in the completely
positive divisible phase covariant dynamics and show that this set
coincides with the set of channels attainable in semigroup phase
covariant dynamics. We also construct new examples of eternally
indivisible dynamics with $\gamma_z(t) < 0$ for all $t > 0$ that
is neither unital nor commutative. Using the quantum Sinkhorn
theorem, we for the first time derive a restriction on the
decoherence rates under which the dynamics is positive divisible,
namely, $\gamma_{\pm}(t) \geq 0$, $\sqrt{\gamma_+(t) \gamma_-(t)}
+ 2 \gamma_z(t)
> 0$. Finally, we consider phase covariant convolution master equations
and find a class of admissible memory kernels that guarantee
complete positivity of the dynamical map.
\end{abstract}

\maketitle

\section{Introduction}
Quantum theory has a well defined statistical
structure~\cite{holevo-st}, where quantum states are associated
with density operators $\varrho$ in a Hilbert space ${\cal H}$,
i.e., positive semidefinite operators with unit trace. Hereafter
we consider a finite-dimensional Hilbert space ${\cal H}_d$ of
dimension $d$. We denote the set of bounded (or linear) operators
acting on ${\cal H}_d$ by ${\cal B}({\cal H}_d)$. Provided the
system state is decoupled from its environment at time $t=0$, any
physical evolution of the quantum system is described by a linear
quantum dynamical map $\Phi(t): {\cal B}({\cal H}_d) \to {\cal
B}({\cal H}_d)$, $t \geq 0$, satisfying the properties of complete
positivity (CP), trace preservation, and the initial condition
$\Phi(0) = {\rm Id}$, where ${\rm Id}$ stands for the identity
transformation~\cite{davies-1976,breuer-book,holevo-book,heinosaari-ziman}.
In the Schr\"{o}dinger picture of the system-environment
evolution,
\begin{equation} \label{Phi-U}
\varrho(t) = \Phi(t) [\varrho(0)] = {\rm tr}_{\rm env} \left[ U(t)
\, \varrho(0) \otimes \xi(0) \, U^{\dag}(t) \right],
\end{equation}

\noindent where $\xi(0) \in {\cal B}({\cal H}^{\rm env})$ is the
initial density operator of the environment, ${\rm tr}_{\rm env}$
is a partial trace over environmental degrees of freedom, and
$U(t) \in {\cal B}({\cal H}_d \otimes {\cal H}^{\rm env})$ is a
unitary evolution operator. The dimension of the effective
reservoir for a general dynamics is estimated in
Ref.~\cite{luchnikov-2019}. Eq.~\eqref{Phi-U} automatically
guarantees that $\Phi(t)$ is completely positive, i.e., the map
$\Phi(t) \otimes {\rm Id}_k$ is positive for all identity
transformations ${\rm Id}_k: {\cal B}({\cal H}_k) \to {\cal
B}({\cal H}_k)$, and $\Phi(t)$ is trace preserving, i.e., ${\rm
tr}\big[\Phi(t)[X]\big] = {\rm tr}[X]$ for all $X \in {\cal
B}({\cal H}_d)$.

An important class of quantum dynamical maps is Markov semigroups
$\Phi(t) = e^{t L}$, Ref.~\cite{alicki-book}. Complete positivity
of $\Phi(t)$ forces the generator $L: {\cal B}({\cal H}_d) \to
{\cal B}({\cal H}_d)$ to be of a special
Gorini-Kossakowski-Sudarshan-Lindblad (GKSL)
form~\cite{gks-1976,lindblad-1976}:
\begin{equation} \label{gksl}
L[\varrho] = - i [H, \varrho] + \sum_k \gamma_k \left( A_k \varrho
A_k^{\dag} - \frac{1}{2} \{\varrho, A_k^{\dag} A_k \} \right),
\end{equation}

\noindent where $H = H^{\dag}$, $\gamma_k \geq 0$ are the
decoherence rates, $A_k: {\cal H}_d \to {\cal H}_d$ are jump
operators, $[\cdot,\cdot]$ and $\{\cdot,\cdot\}$ denote the
commutator and anticommutator, respectively. Markov semigroups
turned out to be an adequate description of open quantum dynamics
in the weak-coupling limit~\cite{van-hove-1954,davies-1974}, the
singular-coupling limit~\cite{palmer-1977,gorini-1978}, the
stochastic limit~\cite{accardi-book}, the low-density limit and
monitoring approach for gas
environment~\cite{dumcke-1985,accardi-1991,apv-2002,hornberger-2007,smirne-2010},
the stroboscopic limit in the collision
model~\cite{rau-1963,scarani-2002,luchnikov-2017,filippov-2019}.
However, it is worth mentioning that the Markov semigroup $e^{t
L}$ cannot be exactly reproduced by Eq.~\eqref{Phi-U} with unitary
operator $U(t) = e^{- i H' t}$ unless all $\gamma_k = 0$ or the
spectrum of the system-environment Hamiltonian $H'$ is unbounded
from below~\cite{exner-book}.

Covariance of a quantum dynamical map $\Phi(t)$ with respect to a
unitary representation $g \rightarrow V_g \in {\cal B}({\cal
H}_d)$ of group $G$, $g \in G$, means that there exists a unitary
representation $g \rightarrow W_g$, $g \in G$, in ${\cal H}_d$
such that~\cite{holevo-1993,holevo-1996}
\begin{equation} \label{covariance}
\Phi[V_g \varrho V_g^{\dag}] = W_g \Phi[\varrho] W_g^{\dag}
\end{equation}

\noindent for all $g \in G$ and all density operators $\varrho$.
Covariance implies some particular structure on the dynamical map
$\Phi(t)$~\cite{fk-2019} and, in the special case of the Markov
semigroup, on the generator
$L$~\cite{holevo-1993,holevo-1996,vacchini-2010}.

\begin{figure}
\includegraphics[width=8cm]{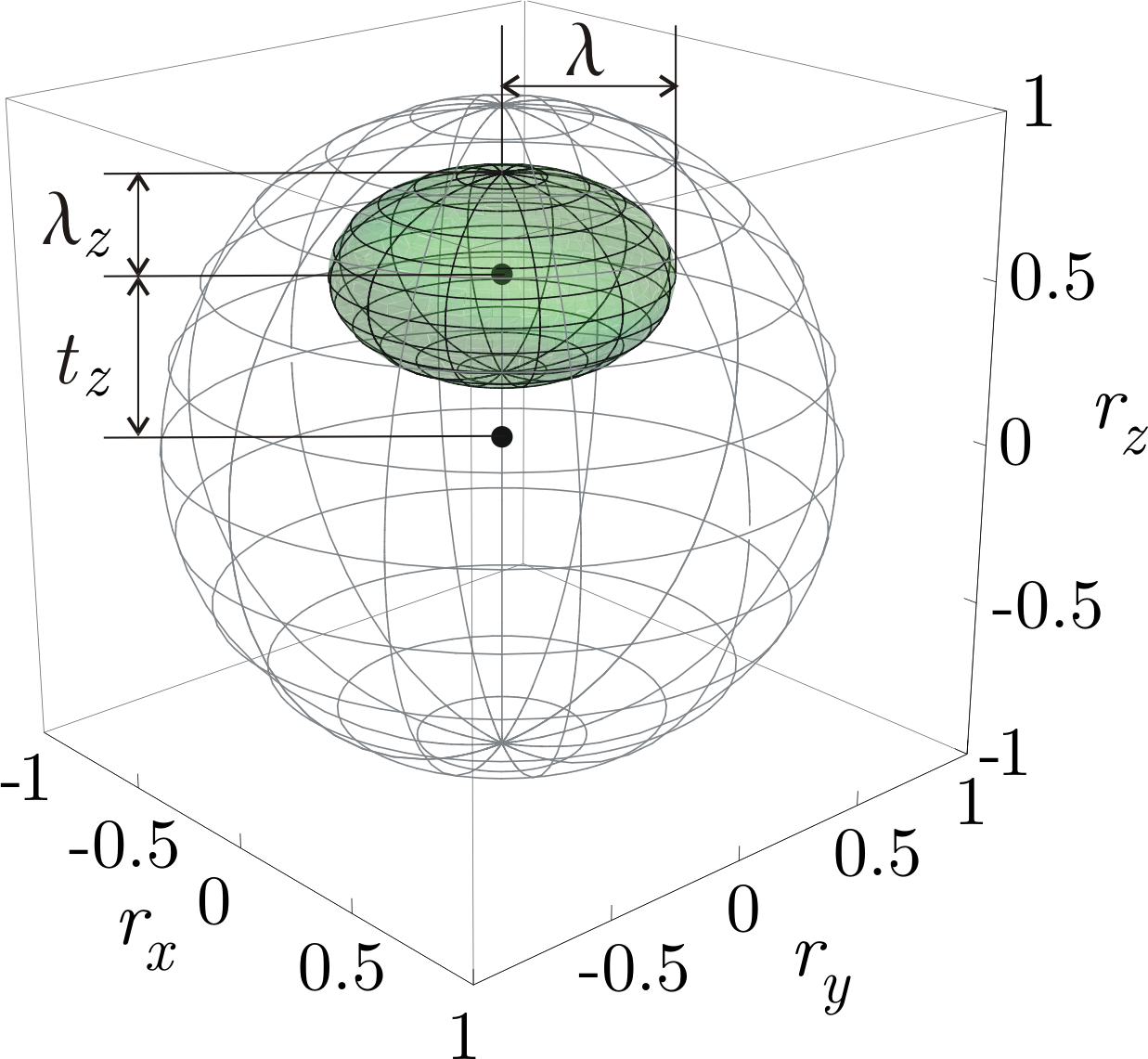}
\caption{\label{figure-1} Bloch ball transformation under the
action of the phase covariant map~\eqref{Phi-channel}.}
\end{figure}

In this paper, we consider phase covariant quantum dynamical maps
for two-level systems (qubits, $d=2$) that satisfy the relation
$\exp( - i \sigma_z \varphi) \Phi[\varrho] \exp( i \sigma_z
\varphi) = \Phi[\exp( - i \sigma_z \varphi) \varrho \exp( i
\sigma_z \varphi)]$ for all real $\varphi$. Hereafter,
$\sigma_x,\sigma_y,\sigma_z \in {\cal B}({\cal H}_2)$ are the
conventional Pauli operators and $\sigma_0 = I$ is the identity
operator on ${\cal H}_2$. Up to an irrelevant transformation
$\varrho \rightarrow \exp( - i \sigma_z \theta) \varrho \exp( i
\sigma_z \theta)$, $\theta \in \mathbb{R}$, the phase covariant
qubit dynamical map $\Phi(t)$
reads~\cite{smirne-2016,lankinen-2016,haase-2018,teittinen-2018,haase-2019}
\begin{equation} \label{Phi-phase-covariant-qubit}
\Phi(t) [\varrho] = \frac{1}{2} \left\{ {\rm tr}[\varrho] \big( I
+ t_z(t) \sigma_z \big) + \lambda(t) {\rm tr}[\sigma_x \varrho]
\sigma_x + \lambda(t) {\rm tr}[\sigma_y \varrho] \sigma_y +
\lambda_z(t) {\rm tr}[\sigma_z \varrho] \sigma_z \right\}
\end{equation}

\noindent and is fully characterized by three real-valued
functions $\lambda(t)$, $\lambda_z(t)$, and $t_z(t)$. The
trace-preservation condition for $\Phi(t)$ is automatically
fulfilled, whereas the complete positivity of $\Phi(t)$ is
equivalent to positivity of the Choi state $\Omega_{\Phi(t)} =
\big( \Phi(t) \otimes {\rm Id} \big) [\ket{\psi_+} \bra{\psi_+}]
\in {\cal B}({\cal H}_4)$, where ${\cal H}_4 \ni \ket{\psi_+} =
\frac{1}{\sqrt{2}} ( \ket{0} \otimes \ket{0} + \ket{1} \otimes
\ket{1} )$, with $\{\ket{0},\ket{1}\}$ being an orthonormal basis
in ${\cal H}_2$, see, e.g., Ref.~\cite{ruskai-2002}. Direct
calculation shows that $\Phi(t)$ is completely positive, and hence
a valid dynamical map, if and only if
\begin{equation} \label{CP}
|\lambda_z(t)| + |t_z(t)| \leq 1 \quad \text{and} \quad 4
\lambda^2(t) + t_z^2(t) \leq [1 + \lambda_z(t)]^2.
\end{equation}

For a fixed time $t \geq 0$ the map $\Phi(t):{\cal B}({\cal H}_2)
\to {\cal B}({\cal H}_2)$ is a phase covariant qubit channel that
we will further denote by $\Phi$ for brevity. $\Phi$ is given by
three real parameters $\lambda$, $\lambda_z$, and $t_z$,
\begin{equation} \label{Phi-channel}
\Phi [\varrho] = \frac{1}{2} \left\{ {\rm tr}[\varrho] \big( I +
t_z \sigma_z \big) + \lambda {\rm tr}[\sigma_x \varrho] \sigma_x +
\lambda {\rm tr}[\sigma_y \varrho] \sigma_y + \lambda_z {\rm
tr}[\sigma_z \varrho] \sigma_z \right\}.
\end{equation}

\noindent The action of $\Phi$ on the set of qubit density
operators has a clear geometrical meaning. Any qubit density
operator $\varrho$ is parameterized by a Bloch vector ${\bf r} \in
\mathbb{R}^3$ inside the Bloch ball $|{\bf r}| \leq 1$, namely,
$\varrho = \frac{1}{2}(I + r_x \sigma_x + r_y \sigma_y + r_z
\sigma_z)$. Action of $\Phi$ on $\varrho$ leads to an affine
transformation: $r_x \to \lambda r_x$, $r_y \to \lambda r_y$, $r_z
\to \lambda_z r_z + t_z$. In other words, the Bloch ball is
contracted into an ellipsoid of revolution with the principal
semi-axes $|\lambda|$, $|\lambda|$, and $|\lambda_z|$ and shifted
by $t_z$ along the $z$-axis, see Fig.~\ref{figure-1}.
Conditions~\eqref{CP} for complete positivity of $\Phi$ are
visualized in Fig.~\ref{figure-2}.

\begin{figure}
\includegraphics[width=8cm]{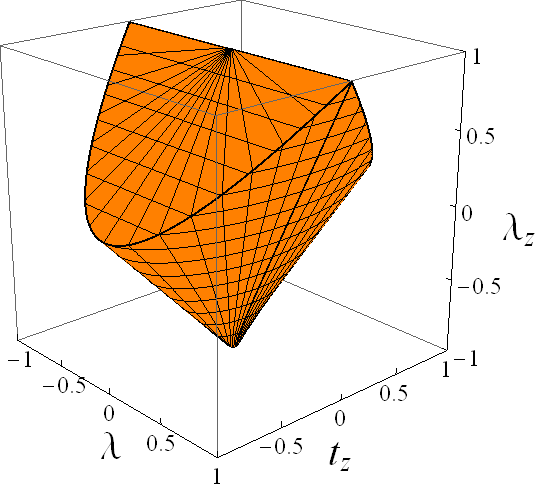}
\caption{\label{figure-2} Region of parameters
$(\lambda,\lambda_z,t_z)$ for which the map~\eqref{Phi-channel} is
completely positive.}
\end{figure}

The important feature of the phase covariant qubit dynamical maps
is that the concatenation of two such maps $\Phi_1(t)$ and
$\Phi_2(t)$ is again phase covariant, however, $\Phi_1(t)\Phi_2(t)
\neq \Phi_2(t) \Phi_1(t)$, i.e., phase covariant qubit dynamical
maps are not commutative in general. In particular,
$\Phi(t)\Phi(s) \neq \Phi(s)\Phi(t)$ for a general phase covariant
dynamics $\Phi(t)$ and different time moments $s$ and $t$.

Much attention has been recently paid to divisibility of quantum
dynamical maps and the study of its relation with non-Markovianity
(see the
reviews~\cite{rivas-2014,breuer-2016,de-vega-2017,benatti-2017,li-2018,pollock-2018,milz-2019,li-2019}).
The notion of divisibility is based on the intermediate propagator
map $\Lambda(t_2,t_1)$ between time moments $t_1 \geq 0$ and $t_2
\geq t_1$ such that $\Phi(t_2) = \Lambda(t_2,t_1) \Phi(t_1)$.
Hereafter we assume that the quantum dynamical map $\Phi(t)$ is
invertible for any $t \geq 0$, i.e., $\lambda(t) \lambda_z(t) \neq
0$ for any finite $t$. Then
\begin{equation} \label{Lambda-intermediate}
\Lambda(t_2,t_1) = \Phi(t_2) \Phi^{-1}(t_1).
\end{equation}

\noindent $\Phi(t)$ is called positive divisible (P-divisible) if
$\Lambda(t_2,t_1)$ is positive for all $t_2 \geq t_1 \geq 0$,
i.e., $\Lambda(t_2,t_1)[X] \geq 0$ for all positive semidefinite
operators $X \ge 0$. Similarly, $\Phi(t)$ is called completely
positive divisible (CP-divisible) if $\Lambda(t_2,t_1)$ is
completely positive for all $t_2 \geq t_1 \geq 0$.

For differentiable quantum dynamical maps $\Phi(t)$ the condition
of CP-divisibility is readily reformulated in terms of the time
dependent generator $L(t) = \frac{d\Phi(t)}{dt} \Phi^{-1}(t)$. In
fact, for the phase covariant qubit dynamical
map~\eqref{Phi-phase-covariant-qubit} we have
\begin{equation} \label{L-generator}
L(t)[\varrho] = \gamma_+(t) \left( \sigma_+ \varrho \sigma_- -
\frac{1}{2} \{\varrho, \sigma_- \sigma_+ \} \right) + \gamma_-(t)
\left( \sigma_- \varrho \sigma_+ - \frac{1}{2} \{\varrho, \sigma_+
\sigma_- \} \right) + \gamma_z(t) \left( \sigma_z \varrho \sigma_z
- \varrho \right),
\end{equation}

\noindent where $\sigma_{\pm} = \frac{1}{2}( \sigma_x \pm i
\sigma_y )$ and the real valued decoherence rates
$\gamma_{\pm}(t)$ and $\gamma_z(t)$ are expressed through
$\lambda(t)$, $\lambda_z(t)$, and $t_z(t)$ as follows:
\begin{equation} \label{gammas-through-lambdas-and-tz}
\gamma_{\pm}(t) = \frac{\lambda_z(t)}{2} \frac{d}{dt} \left(
\frac{1 \pm t_z(t)}{\lambda_z(t)} \right), \quad \gamma_z(t) =
\frac{1}{4\lambda_z(t)} \frac{d\lambda_z(t)}{dt} -
\frac{1}{2\lambda(t)} \frac{d\lambda(t)}{dt} = \frac{1}{4}
\frac{d}{dt} \ln \frac{\lambda_z(t)}{\lambda^2(t)} .
\end{equation}

\noindent Given the time-local generator $L(t)$, the density
matrix evolution is given by the master equation
\begin{equation} \label{master-equation}
\frac{d\varrho(t)}{dt} = L(t)[\varrho(t)],
\end{equation}

\noindent which has a clear physical meaning: the first term in
the dissipator describes energy gain, the second term describes
energy dissipation, and the third term describes pure dephasing.
Therefore, the time-convolutionless master
equation~\eqref{master-equation} is a generalization of the
commonly used decoherence model for spin systems and
superconducting qubits that involves two characteristic times
$T_1$ and $T_2$~\cite{ithier-2005,chernyavskiy}.

Physically, the propagator $\Lambda(t+dt,t)$ for an infinitesimal
time interval $dt$ reads $e^{L(t) dt}$. Therefore,
$\Lambda(t+dt,t)$ is completely positive if and only if $L(t)$ is
a time-local version of the GKSL generator, i.e., $\gamma_{\pm}(t)
\geq 0$ and $\gamma_z(t) \geq 0$. On the other hand, if all
infinitesimal propagators are completely positive, then
$\Lambda(t_2,t_1)$ is completely positive for all $t_2 \geq t_1
\geq 0$. Hence, $\Phi(t)$ is CP-divisible if and only if
$\gamma_{\pm}(t) \geq 0$ and $\gamma_z(t) \geq 0$.

It was recently noticed in Ref.~\cite{haase-2019} that the
population $p(t) := \frac{1}{2}(1 + {\rm tr}[\varrho(t)\sigma_z])$
can be a non-monotonic function of time $t$ in a CP-divisible
phase covariant dynamics for some initial states $\varrho(0)$. In
this paper, we revisit this issue and demonstrate a CP-divisible
phase covariant dynamics $\Phi(t)$ such that the population $p(t)$
is not monotonic for \textit{all} possible initial states
$\varrho(0)$.

The rates $\gamma_{\pm}(t)$ and $\gamma_z(t)$ can temporarily get
negative values without violating complete positivity of
$\Phi(t)$. This corresponds to a CP-indivisible dynamics, which is
one of possible approaches to define
non-Markovianity~\cite{wolf-prl-2008,rivas-2010,hall-2014}. In
fact, an inverse relation to~\eqref{gammas-through-lambdas-and-tz}
is
\begin{eqnarray}
&& \lambda(t) = \exp\left[-\frac{1}{2} \Gamma_+(t) -\frac{1}{2}
\Gamma_-(t) - 2 \Gamma_z(t) \right], \quad \lambda_z(t) =
\exp\left[-\Gamma_+(t) - \Gamma_-(t) \right], \label{lambdax-trough-gammas}\\
&& t_z(t) = \exp\left[-\Gamma_+(t) - \Gamma_-(t) \right] \int_0^t
\left[ \gamma_+(t') - \gamma_-(t') \right] \exp[ \Gamma_+(t') +
\Gamma_-(t')] dt', \label{tz-trough-gammas}
\end{eqnarray}

\noindent where $\Gamma_{\pm}(t) = \int_0^t \gamma_{\pm}(t') dt'$
and $\Gamma_z(t) = \int_0^t \gamma_z (t') dt'$. The only
restriction on the rates $\gamma_{\pm}(t)$ and $\gamma_z(t)$ is
that $\lambda(t)$, $\lambda_z(t)$, and $t_z(t)$ must satisfy
inequalities~\eqref{CP}, so the rates $\gamma_{\pm}(t)$ and
$\gamma_z(t)$ can temporarily become negative. Moreover, the rate
$\gamma_z(t)$ can be negative for all $t > 0$, which corresponds
to an eternal CP-indivisible
dynamics~\cite{hall-2014,megier-2016,fpmz-2017}. Previously, the
examples of eternal CP-indivisible dynamics were constructed only
in the case of unital qubit dynamical maps that are
commutative~\cite{hall-2014,megier-2016,fpmz-2017}. In this paper,
we construct an extended class of eternal CP-indivisible dynamics
for non-commutative phase covariant qubit maps.

As the positivity of a linear map is much more difficult to
characterize as compared to complete
positivity~\cite{bengtsson-2006}, it is not surprising that the
necessary and sufficient conditions for P-divisibility of the
phase covariant qubit dynamics have remained unknown. In this
paper, we fill this gap and find a criterion of P-divisibility in
terms of decoherence rates $\gamma_{\pm}(t)$ and $\gamma_z(t)$. A
key tool in this study is the quantum Sinkhorn
theorem~\cite{gurvits-2004,aubrun-2015,aubrun-2017,fm-2017,ffk-2017,filippov-romp-2018,filippov-jms-2019}
that allows to characterize the region of parameters $\lambda$,
$\lambda_z$, and $t_z$ such that Eq.~\eqref{Phi-channel} defines a
positive map. P-divisibility of qubit dynamical maps implies many
interesting properties, for instance, monotonically decreasing
distinguishability of quantum states~\cite{blp}, monotonically
decreasing volume of accessible states~\cite{lorenzo-2013}, and
monotonically decreasing relative
entropy~\cite{muller-hermes-2017}.

The dynamics of a $d$-dimensional open quantum system can be
alternatively described by means of the Nakajima-Zwanzig
projective
techniques~\cite{breuer-book,nakajima-1958,zwanzig-1960,chruscinski-2019}
leading to an integro-differential master equation of the form
\begin{equation} \label{NZ}
\frac{d \varrho(t)}{dt} = \int_0^t K(t') [\varrho(t - t')] dt'
\end{equation}

\noindent with the memory kernel $K(t') : {\cal B}({\cal H}_d) \to
{\cal B}({\cal H}_d)$. In terms of the dynamical maps,
Eq.~\eqref{NZ} reads $\frac{d\Phi(t)}{dt} = \int_0^t K(t') \Phi(t
- t') dt'$. It is an open question what memory kernels $K(t')$
define a legitimate quantum dynamics $\Phi(t)$ in general. The
admissible memory kernels are characterized for Pauli dynamical
maps~\cite{chruscinski-2015} and quantum semi-Markov
processes~\cite{chruscinski-2016,chruscinski-2017,chruscinski-2-2017}.
In this paper, we give sufficient conditions for admissible memory
kernels corresponding to legitimate phase covariant qubit
dynamics. It facilitates the analysis of divisibility of such
dynamical maps by subjecting them to time
deformations~\cite{fc-2018}.

The paper is organized as follows. In
Section~\ref{section-attainable}, we characterize
channels~\eqref{Phi-channel} that are attainable in the phase
covariant CP divisible dynamics and semigroup dynamics. In
Section~\ref{section-non-monotonicity}, we revisit the result of
Ref.~\cite{haase-2019} and construct a phase covariant process
with non-negative rates $\gamma_{\pm}(t)$ and $\gamma_z(t) \geq 0$
such that the population $p(t)$ is not monotonic for all initial
states. In Section~\ref{section-eternal-indivisible}, we construct
two families of non-unital eternal CP-indivisible phase covariant
processes, one of which is commutative and the other is not. In
Section~\ref{section-P-divisibility}, we use the quantum Sinkhorn
theorem and find parameters $\lambda,\lambda_z,t_z$ for which
Eq.~\eqref{Phi-channel} defines a positive map. We further use
this result and characterize positive divisible dynamical
maps~\eqref{Phi-phase-covariant-qubit}. In
Section~\ref{section-memory-kernel}, we find a class of admissible
memory kernels $K(t)$ leading to legitimate phase covariant qubit
dynamics. In Section~\ref{section-conclusions}, brief conclusions
are given.

\section{Channels attainable in semigroup dynamics and CP-divisible
dynamics}\label{section-attainable}

Following the classification of
Refs.~\cite{wolf-prl-2008,davalos,puchala}, consider a class of
quantum channels $\mathsf{C}^L$ attainable in at least one phase
covariant semigroup dynamics, i.e.,
\begin{equation*}
\mathsf{C}^L = \overline{ \left\{ e^{Lt} \ \vert \ t \geq 0
\text{~and~} L \text{~has the form~} \eqref{L-generator}
\text{~with constant rates~} \gamma_{\pm}, \gamma_z \geq 0
\right\} },
\end{equation*}

\noindent where the overline denotes the closure in operator norm
for Choi matrices.

\begin{proposition} \label{prop-C-L}
$\mathsf{C}^L$ consists of phase covariant qubit
channels~\eqref{Phi-channel} with $|\lambda_z| + |t_z| \leq 1$, $4
\lambda^2 + t_z^2 \leq (1 + \lambda_z)^2$, $\lambda \geq 0$, and
$\lambda_z \geq \lambda^2$.
\end{proposition}

\begin{proof}
The relations~\eqref{lambdax-trough-gammas} and
\eqref{tz-trough-gammas} imply that in the semigroup dynamics
\begin{equation*}
\lambda(t) = e^{-\frac{1}{2} (\gamma_+ + \gamma_- + 4 \gamma_z)
t}, \quad \lambda_z(t) = e^{-(\gamma_+ + \gamma_-) t }, \quad
t_z(t) = \frac{\gamma_+ - \gamma_-}{\gamma_+ + \gamma_-} \left[ 1
- e^{-(\gamma_+ + \gamma_-)t}\right]
\end{equation*}

\noindent and necessarily satisfy the claimed conditions. To see
the other direction, suppose $|\lambda_z| + |t_z| \leq 1$, $4
\lambda^2 + t_z^2 \leq (1 + \lambda_z)^2$, $\lambda > 0$, and $1
> \lambda_z > \lambda^2$, so that the channel $\Phi$ with parameters
$\lambda,\lambda_z,t_z$ can be expressed as $\Phi = e^{L}$, where
\begin{equation*}
\gamma_{\pm} = \frac{(1-\lambda_z \pm t_z)(- \ln
\lambda_z)}{2(1-\lambda_z)} \geq 0, \quad \gamma_z = \frac{1}{4}
\ln \frac{\lambda_z}{\lambda^2} \geq 0.
\end{equation*}

\noindent The non-strict inequalities for $\lambda$, $\lambda_z$,
and $t_z$ then follow from the closure procedure.
\end{proof}

\begin{remark} \label{remark-1}
Restoring rotations around $z$-axis of the Bloch ball ($\varrho
\rightarrow e^{-i \sigma_z \omega t} \varrho e^{i \sigma_z \omega
t}$, $\omega \in \mathbb{R}$) into a general phase covariant
dynamics, we get an extra term $-i \omega [\sigma_z, \varrho]$ in
$L[\varrho]$. Whenever $\sin \omega t = \pm 1$, this results in
the change $\lambda(t) \rightarrow - \lambda(t)$ in
Eq.~\eqref{Phi-phase-covariant-qubit}. Hence, the region of
parameters $\lambda$, $\lambda_z$, $t_z$ attainable in a general
phase covariant semigroup dynamics is twice larger, namely,
$|\lambda_z| + |t_z| \leq 1$, $4 \lambda^2 + t_z^2 \leq (1 +
\lambda_z)^2$, and $\lambda_z \geq \lambda^2$.
\end{remark}

\begin{figure}
\includegraphics[width=8cm]{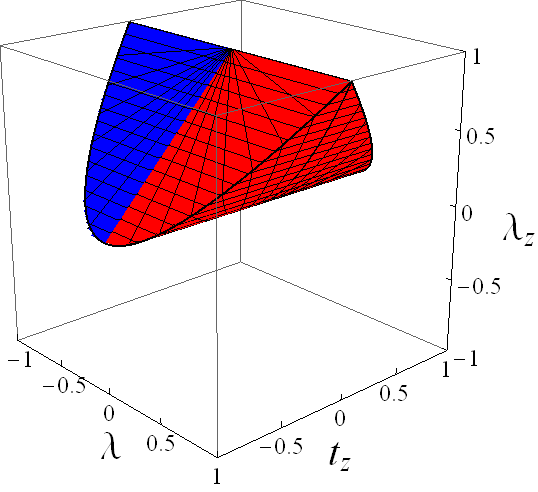}
\caption{\label{figure-3} Region of parameters
$(\lambda,\lambda_z,t_z)$ for which the
channel~\eqref{Phi-channel} can be attained as a result of
entirely dissipative semigroup dynamics (red) and its extension
due to the coherent part $-i \omega [\sigma_z, \cdot ]$ in the
generator $L$ (blue).}
\end{figure}

The class $\mathsf{C}^L$ is illustrated in Fig.~\ref{figure-3}.
Every quantum channel within the class $\mathsf{C}^L$ is a convex
mixture of a dephasing channel in the basis $\{\ket{0},\ket{1}\}$,
an amplitude damping channel with a stationary point
$\ket{0}\bra{0}$, and an amplitude damping channel with a
stationary point $\ket{1}\bra{1}$. Generalized amplitude damping
semigroups $\Phi(t) = \exp\left[ \gamma_+ t \left( \sigma_+
\varrho \sigma_- - \frac{1}{2} \{\varrho, \sigma_- \sigma_+ \}
\right) + \gamma_- t \left( \sigma_- \varrho \sigma_+ -
\frac{1}{2} \{\varrho, \sigma_+ \sigma_- \} \right) \right] $ are
ultimately CP-divisible, which means that an infinitesimal
perturbation in the trajectory $(\lambda(t),\lambda_z(t),t_z(t))$
of such a dynamical map may break the CP-divisibility
property~\cite{fpmz-2017}.

Extending the classification of
Refs.~\cite{wolf-prl-2008,davalos}, consider a class of quantum
channels $\mathsf{C}_{\rm ph~cov}^{\rm CP}$ attainable in a
CP-divisible phase covariant dynamics, i.e.,
\begin{equation*}
\mathsf{C}_{\rm ph~cov}^{\rm CP} = \overline{ \left\{
T_{\leftarrow} \exp\left( \int_0^{t} L(t') dt' \right) \ \vert \ t
\geq 0 \text{~and~} L(t) \text{~has the form~} \eqref{L-generator}
\text{~with~} \gamma_{\pm}(t), \gamma_z(t) \geq 0 \right\} },
\end{equation*}

\noindent where $T_{\leftarrow}$ is the Dyson time ordering
operator.

\begin{proposition} \label{prop-C-CP}
$\mathsf{C}_{\rm ph~cov}^{\rm CP} = \mathsf{C}^L$.
\end{proposition}

\begin{proof}
Obviously, $\mathsf{C}^L \subset \mathsf{C}_{\rm ph~cov}^{\rm CP}$
as a particular case of time-independent decoherence rates. Let us
prove that $\mathsf{C}_{\rm ph~cov}^{\rm CP} \subset
\mathsf{C}^L$. Note that the dynamics with the generator $L(t)$ of
the form~\eqref{L-generator} allows only non-negative values of
$\lambda(t)$, see Eq.~\eqref{lambdax-trough-gammas}. Consider a
valid quantum channel $\Phi'$ of the form~\eqref{Phi-channel} with
parameter $\lambda \geq 0$. As $\Phi'$ is completely positive, the
conditions $|\lambda_z| + |t_z| \leq 1$ and $4 \lambda^2 + t_z^2
\leq (1 + \lambda_z)^2$ are automatically satisfied. By
Proposition~\ref{prop-C-L}, if $\Phi' \not \in \mathsf{C}^L$, then
$\lambda_z < \lambda^2$. Suppose $\Phi' \not \in \mathsf{C}^L$ and
$\Phi'$ is attained by a phase covariant
dynamics~\eqref{Phi-phase-covariant-qubit} with some time
dependent rates $\gamma_{\pm}(t)$ and $\gamma_z(t)$, i.e., $\Phi'
= \Phi(t_0)$ for some $t_0 \geq 0$. As $\lambda_z < \lambda^2$,
Eq.~\eqref{gammas-through-lambdas-and-tz} implies that
$\gamma_z(t_0) < 0$, i.e., $\Phi(t)$ is not CP divisible.
Therefore, a CP divisible phase covariant dynamics with the
time-local generator~\eqref{L-generator} results in the dynamical
maps $\Phi(t)$ such that $\Phi(t_0) \in \mathsf{C}^L$ for any
fixed $t_0 \geq 0$, i.e., $\mathsf{C}_{\rm ph~cov}^{\rm CP}
\subset \mathsf{C}^L$.
\end{proof}

\begin{remark} \label{remark-2}
As in Remark~\ref{remark-1}, restoring rotations around $z$-axis
of the Bloch ball into a general phase covariant dynamics, we get
the twice larger region of parameters $\lambda$, $\lambda_z$,
$t_z$ for channels $\Phi$ attainable by CP-divisible phase
covariant dynamics, namely, $|\lambda_z| + |t_z| \leq 1$, $4
\lambda^2 + t_z^2 \leq (1 + \lambda_z)^2$, and $\lambda_z \geq
\lambda^2$.
\end{remark}

\begin{remark} \label{remark-3}
If a phase covariant qubit channel $\Phi \neq {\rm Id}$ is
obtained as a result of a semigroup dynamics, i.e., $\Phi = e^{L
t_0}$ for some $t_0 > 0$, then any channel from the semigroup
$e^{Lt}$, $t \geq 0$, is phase covariant. However, if a phase
covariant qubit channel $\Phi \neq {\rm Id}$ is obtained as a
result of a general qubit CP-divisible dynamics $\Theta(t)$, i.e.,
$\Phi = \Theta(t_0)$ for some $t_0 > 0$, then $\Theta(t)$ does not
have to be phase covariant for all $t \geq 0$. In fact, consider a
phase covariant qubit channel $\Phi'[\varrho] = \frac{1}{2} \left(
{\rm tr}[\varrho] I - {\rm tr}[\sigma_z \varrho] \sigma_z \right)$
such that $\Phi' \not\in \mathsf{C}^L = \mathsf{C}_{\rm
ph~cov}^{\rm CP}$. Note that $\Phi'[\varrho] = \sigma_x
\Phi[\varrho] \sigma_x$, where $\Phi[\varrho] = \frac{1}{2} \left(
{\rm tr}[\varrho] I + {\rm tr}[\sigma_z \varrho] \sigma_z
\right)$, i.e., $\Phi \in \mathsf{C}^L = \mathsf{C}_{\rm
ph~cov}^{\rm CP}$. Therefore, the channel $\Phi'$ can be obtained
as a result of a CP-divisible dynamics, however, the intermediate
transformation $\varrho \to \sigma_x \varrho \sigma_x$ is not
phase covariant. This example shows that $\Phi' \not \in
\mathsf{C}_{\rm ph~cov}^{\rm CP}$ but $\Phi' \in \mathsf{C}^{\rm
CP}$, where
\begin{equation*}
\mathsf{C}^{\rm CP} = \overline{ \left\{ T_{\leftarrow} \exp\left(
\int_0^{t} L(t') dt' \right) \ \vert \ t \geq 0 \text{~and~} L(t)
\text{~has a time-dependent form~} \eqref{gksl} \text{~with~}
\gamma_{k}(t) \geq 0 \right\} }.
\end{equation*}

\noindent By using the results of
Refs.~\cite{wolf-prl-2008,davalos} and the explicit form of the
quantum Sinkhorn theorem for phase covariant qubit channels
(Proposition~\ref{prop-P}) as well as taking into account possible
unitary rotations of the Bloch ball, we conclude that a
channel~\eqref{Phi-channel} with parameters
$\lambda,\lambda_z,t_z$ belongs to the class $\mathsf{C}^{\rm CP}$
if and only if $|\lambda_z| + |t_z| \leq 1$, $4 \lambda^2 + t_z^2
\leq (1 + \lambda_z)^2$, $\lambda_z \geq \lambda^2$ or
$|\lambda_z| + |t_z| \leq 1$, $4 \lambda^2 + t_z^2 \leq (1 +
\lambda_z)^2$, $\lambda = 0$.
\end{remark}

\section{Non-monotonicity of population in CP-divisible dynamics}
\label{section-non-monotonicity}

\begin{proposition} \label{prop-non-monotonicity}
There exists a CP-divisible phase covariant qubit dynamical map
$\Phi(t)$ such that the population $p(t)$ is non-monotonic for any
initial state $\varrho(0)$.
\end{proposition}

\begin{proof}
Let $\lambda(t) = e^{-\nu t}$, $\lambda_z(t) = e^{- 2 \nu t}$, and
$t_z(t) = \frac{2\nu}{\sqrt{4\nu^2 + \omega^2}} \sin \omega t$,
where $\nu, \omega > 0$. Then the population reads
\begin{equation*} \label{population}
p(t) = \frac{1}{2} \Big[ 1 + t_z(t) + \lambda_z(t) {\rm
tr}[\varrho(0)\sigma_z] \Big] = \frac{1}{2} \left[ 1 +
\frac{2\nu}{\sqrt{4\nu^2 + \omega^2}} \sin \omega t + e^{- 2 \nu
t} {\rm tr}[\varrho(0)\sigma_z] \right]
\end{equation*}

\noindent and clearly has a non-monotonic behaviour for all
initial density operators $\varrho(0)$ if $t > \frac{1}{2\nu} \ln
\frac{\sqrt{4\nu^2 + \omega^2}}{2\nu}$.

To guarantee that $\Phi(t)$ is a valid CP-divisible dynamical map
it suffices to check that the rates $\gamma_{\pm}(t)$ and
$\gamma_z(t)$ are non-negative for any $t \geq 0$. Substituting
our particular choice for $\lambda(t)$, $\lambda_z(t)$, and
$t_z(t)$ into Eq.~\eqref{gammas-through-lambdas-and-tz}, we get
\begin{equation*}
\gamma_{\pm}(t) = \nu \pm \frac{\nu}{\sqrt{4\nu^2 + \omega^2}}
\left( 2 \nu \sin \omega t + \omega \cos \omega t \right) \geq 0,
\quad \gamma_z(t) = 0.
\end{equation*}

\noindent Therefore $\Phi(t)$ is a valid quantum dynamical map
enjoying the CP-divisibility property.
\end{proof}

The construction used in the proof of
Proposition~\ref{prop-non-monotonicity} has a clear physical
meaning too. As $\gamma_z(t) = 0$ for all time moments $t \geq 0$,
the master equation~\eqref{master-equation} with the
generator~\eqref{L-generator} defines a generalized amplitude
damping dynamics~\cite{nielsen}, where the time-local stationary
state changes in time. Note that the population
oscillations~\eqref{population} do not decay in time. The
peak-to-peak amplitude $\max_{t \geq t_0} p(t) - \min_{t \geq t_0}
p(t) \geq \frac{2\nu}{\sqrt{4\nu^2 + \omega^2}}$  for all $t_0
\geq 0$ and can be arbitrarily close to 1 if $\omega \ll \nu$.

\section{Eternal CP-indivisible dynamics} \label{section-eternal-indivisible}

In this section, we construct a one-parameter family of phase
covariant qubit dynamical maps $\{\Phi_a(t)\}_{|a| < 1}$ such that
the intermediate map $\Lambda(t_2,t_1)$ is not completely positive
for any $t_2 > t_1 > 0$. Since $\Phi_a(t)$ is non-unital if $a
\neq 0$, our construction provides a family of non-unital eternal
CP-indivisible dynamical maps.

\begin{proposition} \label{prop-eternal-commutative}
A phase covariant qubit dynamical map $\{\Phi_a(t)\}$ of the
form~\eqref{Phi-phase-covariant-qubit} with
\begin{equation*}
\lambda(t) = \frac{1}{2} \sqrt{(1 + e^{-2\nu t})^2 - a^2 (1 -
e^{-2\nu t})^2}, \quad \lambda_z(t) = e^{-2\nu t}, \quad t_z(t) =
a (1 - e^{-2\nu t}), \quad \nu>0,
\end{equation*}

\noindent is eternal CP indivisible for all real $a$ satisfying
$|a| < 1$.
\end{proposition}

\begin{proof}
$\{\Phi_a(t)\}$ is a valid quantum dynamical map if $|a| < 1$
because the conditions~\eqref{CP} are fulfilled. A direct
calculation by Eq.~\eqref{gammas-through-lambdas-and-tz} yields
\begin{equation*}
\gamma_{\pm}(t) = \nu (1 \pm a), \quad \gamma_z(t) = - \frac{\nu
(1-a^2) \sinh 2 \nu t}{2 [1+a^2 + (1-a^2) \cosh 2 \nu t ]}.
\end{equation*}

\noindent If $|a|<1$, then $\gamma_z(t) < 0$ for all $t>0$, so
$\Phi_a(t)$ is eternal CP-indivisible.
\end{proof}

A feature of the dynamical map $\Phi_a(t)$ in
Proposition~\ref{prop-eternal-commutative} is that it is
commutative, i.e., $\Phi_a(t) \Phi_a(s) = \Phi_a (s) \Phi_a(t)$,
because the ratio $\frac{\gamma_+(t)}{\gamma_-(t)}$ is constant in
time. The following proposition shows that there also exist
non-commutative eternal CP-indivisible phase covariant qubit
processes.

\begin{proposition} \label{prop-eternal-non-commutative}
A phase covariant qubit dynamical map $\{\Phi_b(t)\}$ of the
form~\eqref{Phi-phase-covariant-qubit} with
\begin{equation*}
\lambda(t) = \frac{1}{2} \sqrt{(1 + e^{-2\nu t})^2 - b^2 e^{-2 \nu
t} (1 - e^{-2\nu t})^2}, \quad \lambda_z(t) = e^{-2\nu t}, \quad
t_z(t) = b e^{-\nu t} (1 - e^{-2\nu t}), \quad \nu>0,
\end{equation*}

\noindent is non-commutative and eternal CP-indivisible for all
real $b$ satisfying $0 < |b| \leq 1$.
\end{proposition}

\begin{proof}
$\{\Phi_b(t)\}$ is a valid quantum dynamical map if $|b| \leq 1$
because the conditions~\eqref{CP} are fulfilled. A direct
calculation by Eq.~\eqref{gammas-through-lambdas-and-tz} yields
\begin{equation*}
\gamma_{\pm}(t) = \nu \left( 1 \pm b e^{-2 \nu t} \cosh \nu t
\right), \quad \gamma_z(t) = - \frac{\nu (1 - e^{-2 \nu t})(e^{3
\nu t} \cosh \nu t - b^2)}{4 [e^{2 \nu t} \cosh^2 \nu t - b^2
\sinh^2 \nu t]}.
\end{equation*}

\noindent If $|b| \leq 1$, then $\gamma_z(t) < 0$ for all $t>0$,
so $\Phi_b(t)$ is eternal CP indivisible. Moreover, since $ t_z(t)
[1 - \lambda_z(s)] \neq t_z(s) [1 - \lambda_z(t)]$ for all $s > t
> 0$ and $b \neq 0$, the map $\Phi_b(t)$ is non-commutative for all $b$ satisfying $0 < |b|
\leq 1$.
\end{proof}

Trajectories of the processes $\Phi_a(t)$ and $\Phi_b(t)$ in the
parameter space $(\lambda,\lambda_z,t_z)$ belong to the surface of
the body depicted in Fig.~\ref{figure-2}. Note that the
constructed dynamical maps $\Phi_a(t)$ and $\Phi_b(t)$ reduce to
the known eternal CP-indivisible unital qubit
dynamics~\cite{hall-2014,megier-2016,fpmz-2017} if $a=0$ and
$b=0$, respectively.

\section{Positivity and positive divisibility} \label{section-P-divisibility}

Not only completely positive maps but also positive maps have
attracted some attention recently for description of quantum
systems and their dynamics~\cite{shaji-2005}. For this reason we
characterize a region of parameters $\lambda,\lambda_z,t_z$ within
which the map~\eqref{Phi-channel} is positive, i.e., $\Phi[X] \geq
0$ for all $X \geq 0$.

\begin{proposition} \label{prop-P}
Eq.~\eqref{Phi-channel} defines a positive map if and only if
\begin{equation} \label{positivity-conditions}
\left\{ \begin{array}{c}
  |\lambda_z| + |t_z| \leq 1, \\
  \quad 2|\lambda| \leq \sqrt{(1+\lambda_z)^2 - t_z^2} +
\sqrt{(1-\lambda_z)^2 - t_z^2}, \\
  \quad 4|\lambda_z| \leq \left[ \sqrt{(1+\lambda_z)^2 - t_z^2} +
\sqrt{(1-\lambda_z)^2 - t_z^2} \right]^2. \\
\end{array} \right.
\end{equation}
\end{proposition}

\begin{proof}
Note that positivity of $\Phi$ is equivalent to condition
$\Phi[\varrho] \geq 0$ for all density matrices $\varrho$. We
consider two cases.

Suppose $|\lambda_z| + |t_z| = 1$. The image of the Bloch ball
under map~\eqref{Phi-channel} is an ellipsoid of revolution, which
has a common point with the Bloch sphere either at the north pole
or the south pole. Geometrically, the image ellipsoid is a subset
of the Bloch ball if and only if $\lambda^2 \leq |\lambda_z|$.

Suppose $|\lambda_z| + |t_z| < 1$, then there exist non-degenerate
operators $A$ and $B$ such that the map $\varrho \to
\Upsilon[\varrho] = A \Phi[B \varrho B^{\dag}] A^{\dag}$ is
unital, i.e.,
\begin{equation*}
\Upsilon[\varrho] = \frac{1}{2} \left( {\rm tr}[\varrho] I +
\widetilde{\lambda}_x {\rm tr}[\sigma_x \varrho]\sigma_x +
\widetilde{\lambda}_y {\rm tr}[\sigma_y \varrho]\sigma_y +
\widetilde{\lambda}_z {\rm tr}[\sigma_z \varrho]\sigma_z \right).
\end{equation*}

The explicit form of operators $A$ and $B$ as well as real numbers
$\widetilde{\lambda}_x,\widetilde{\lambda}_y,\widetilde{\lambda}_z$
are derived in Refs.~\cite{fm-2017,ffk-2017,filippov-romp-2018}.
In particular,
\begin{equation} \label{tilde-lambda}
\widetilde{\lambda}_x = \widetilde{\lambda}_y =
\frac{2\lambda}{\sqrt{(1+\lambda_z)^2 - t_z^2} +
\sqrt{(1-\lambda_z)^2 - t_z^2}}, \quad \widetilde{\lambda}_z =
\frac{4 \lambda_z}{\left[ \sqrt{(1+\lambda_z)^2 - t_z^2} +
\sqrt{(1-\lambda_z)^2 - t_z^2} \right]^2}.
\end{equation}

\noindent As operators $A$ and $B$ are non-degenerate,
$\Phi[\varrho] = A^{-1} \Upsilon[B^{-1} \varrho (B^{-1})^{\dag}]
(A^{-1})^{\dag}$. Therefore, $\Phi[\varrho] \geq 0$ if and only if
$\Upsilon$ is a positive map, i.e.,
$|\widetilde{\lambda}_x|,|\widetilde{\lambda}_y|,|\widetilde{\lambda}_z|
\leq 1$.

Finally, the conditions~\eqref{positivity-conditions} summarize
the two cases altogether.

\end{proof}

\begin{remark}
The relation $\Upsilon[\varrho] = A \Psi[B \varrho B^{\dag}]
A^{\dag}$ between a unital trace preserving map $\Upsilon$ and a
strictly positive map $\Psi$ (i.e., $\Psi[X] > 0$ for all $X \geq
0$) is known as the quantum Sinkhorn theorem and reviewed in
Refs.~\cite{gurvits-2004,aubrun-2015,aubrun-2017,fm-2017,ffk-2017,filippov-romp-2018,filippov-jms-2019}.
\end{remark}

\begin{figure}
\includegraphics[width=8cm]{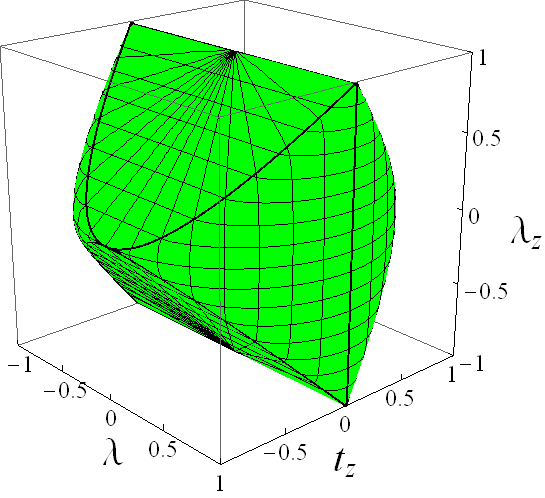}
\caption{\label{figure-4} Region of parameters
$(\lambda,\lambda_z,t_z)$ for which the map~\eqref{Phi-channel} is
positive.}
\end{figure}

Conditions~\eqref{positivity-conditions} are visualized in
Fig.~\ref{figure-4}. The
inequalities~\eqref{positivity-conditions} can be alternatively
reformulated in terms of the maximum distance between the center
of the Bloch ball and a point in the image ellipsoid. Namely,
\begin{equation} \label{r-max}
r_{\max}(\Phi) = \max_{\varrho, \, {\bf n} \in {\mathbb R}^3 \, :
\, |{\bf n}|=1} \frac{1}{2} {\rm tr} \Big[ (n_x \sigma_x + n_y
\sigma_y + n_z \sigma_z) \Phi[\varrho] \Big] = \left\{
\begin{array}{ll}
  |\lambda_z| + |t_z| & \text{if~}|\lambda| \leq |\lambda_z|, \\
  |\lambda| \sqrt{1 + \frac{t_z^2}{\lambda^2 - \lambda_z^2}} & \text{if~} |\lambda| > |\lambda_z|. \\
\end{array} \right.
\end{equation}

\noindent The map $\Phi$ is positive if and only if
$r_{\max}(\Phi) \leq 1$.

Proposition~\ref{prop-P} is a powerful tool that can be also
applied to the intermediate map~\eqref{Lambda-intermediate}. This
results in the following characterization of P-divisible phase
covariant qubit dynamical maps.

\begin{proposition} \label{prop-P-divisibility}
The master equation $\frac{d\varrho(t)}{dt} = L(t)[\varrho(t)]$
with the generator~\eqref{L-generator} defines a P-divisible
dynamical map $\Phi(t)$ if the decoherence rates satisfy
\begin{equation} \label{p-divisibility-conditions}
\gamma_{\pm}(t) \geq 0 \quad \text{and} \quad \sqrt{\gamma_+(t)
\gamma_-(t)} + 2 \gamma_z(t) > 0.
\end{equation}

\noindent $\Phi(t)$ is not P-divisible if either (i)
$\gamma_{+}(t) < 0$, or (ii) $\gamma_{-}(t) < 0$, or (iii)
$\gamma_{\pm}(t) \geq 0$ and $\sqrt{\gamma_+(t) \gamma_-(t)} + 2
\gamma_z(t) < 0$.
\end{proposition}

\begin{proof}
$\Phi(t)$ is P-divisible if and only if the infinitesimal
propagator $\Lambda(t+dt,t)$ is positive for all $t \geq 0$. Since
$\Lambda(t+dt,t) = {\rm Id} + L(t) dt + o(dt)$, we conclude that
$\Lambda(t+dt,t)$ is positive if ${\rm Id} + L(t) dt$ is strictly
positive. The map ${\rm Id} + L(t) dt$ is phase covariant and has
the form of Eq.~\eqref{Phi-channel} with parameters
\begin{equation*}
\lambda'(t) = 1 - \frac{1}{2} (\gamma_+ + \gamma_- + 4 \gamma_z )
dt, \quad \lambda_z'(t) = 1 - (\gamma_+ + \gamma_- ) dt, \quad
t_z'(t) = (\gamma_+ - \gamma_-) dt.
\end{equation*}

\noindent Substituting these expressions for
$\lambda,\lambda_z,t_z$ in a strict version of
inequalities~\eqref{positivity-conditions} and keeping the terms
up to the first order of $dt>0$, we get the equivalent conditions
$\gamma_+(t) + \gamma_-(t) - |\gamma_+(t) - \gamma_-(t)| > 0$ and
$\sqrt{\gamma_+(t) \gamma_-(t)} + 2 \gamma_z(t) > 0$. In turn, the
first inequality is equivalent to $\gamma_{\pm}(t) > 0$. If
$\gamma_{+}(t) = 0$, $\gamma_{-}(t) \geq 0$, $\gamma_z(t) \geq 0$
or $\gamma_{-}(t) = 0$, $\gamma_{+}(t) \geq 0$, $\gamma_z(t) \geq
0$, then the map $\Lambda(t+dt,t)$ is completely positive and,
consequently, positive. Therefore,
conditions~\eqref{p-divisibility-conditions} guarantee positivity
of the intermediate map $\Lambda(t+dt,t)$.

Conversely, suppose $\gamma_{+}(t) < 0$ or $\gamma_{-}(t) < 0$,
then $\Lambda(t+dt,t)$ is not positive as at least one of the
operators $\Lambda(t+dt,t)[\frac{1}{2}(I \pm \sigma_z)]$ is not
positive semidefinite. Suppose $\gamma_{\pm}(t) \geq 0$ and
$\sqrt{\gamma_+(t) \gamma_-(t)} + 2 \gamma_z(t) < 0$, then
Eq.~\eqref{r-max} gives $r_{\max}(\Lambda(t+dt,t)) =  1 +
\frac{2[4\gamma_z^2(t) - \gamma_+(t) \gamma_-(t)]}{\gamma_+(t) +
\gamma_-(t) - 4 \gamma_z(t)}dt + o(dt)
> 1$ and $\Lambda(t+dt,t)$ is not positive.
\end{proof}

\begin{remark}
If $\gamma_{\pm}(t) \geq 0$ and $\sqrt{\gamma_+(t) \gamma_-(t)} +
2 \gamma_z(t) = 0$, then one should resort to higher order
expansions of $\Lambda(t+dt,t)$ with respect to $dt$. The second
order expansion yields the condition $\frac{d\gamma_z(t)}{dt} >
\gamma_z(t)[\gamma_+(t) + \gamma_-(t)]$ for P-divisibility of
$\Phi(t)$.
\end{remark}

The derived condition of
P-divisibility~\eqref{p-divisibility-conditions} is stronger than
the condition for monotonically decreasing distinguishability of
quantum states, $\gamma_+ + \gamma_- \geq 0$ and $\gamma_+ +
\gamma_- + 4 \gamma_z \geq 0$, Ref.~\cite{teittinen-2018}. This is
trivial if one of the rates $\gamma_+$ or $\gamma_-$ is negative.
However, even if $\gamma_{\pm} \geq 0$ we have $\gamma_+ +
\gamma_- \geq 2 \sqrt{\gamma_+ \gamma_-}$, so the P-divisibility
condition is stronger than the deceasing distinguishability
condition. In other words, there exists a phase covariant qubit
dynamics such that the trace distance
$\frac{1}{2}\|\Phi(t)[\varrho_1] - \Phi(t)[\varrho_2]\|_1$
monotonically decreases for arbitrary two initial states
$\varrho_1$ and $\varrho_2$, however, the dynamics is not
P-divisible.

\section{Admissible memory kernels} \label{section-memory-kernel}

In this section, we address the following question: what memory
kernels $K(t') : {\cal B}({\cal H}_d) \to {\cal B}({\cal H}_d)$ in
the master equation~\eqref{NZ} define a legitimate (completely
positive and trace preserving) phase covariant qubit dynamics
$\Phi(t)$? We provide a sufficient condition for memory kernels
and illustrate our findings by an example.

The relation between a dynamical map $\Phi(t)$ and the memory
kernel $K(t)$ takes the simplest form if we use the Laplace
transform $F_s = \int_0^{\infty} F(t) e^{-st} dt$. Indeed, the
Laplace transform of the equation $\frac{d \Phi(t)}{dt} = \int_0^t
K(t') [\Phi(t - t')] dt'$ yields $s \Phi_s - {\rm Id} = K_s
\Phi_s$, where we have taken into account that $\Phi(0) = {\rm
Id}$. Therefore, for a phase covariant qubit dynamics with the
memory kernel
\begin{equation} \label{K-kernel}
K(t)[\varrho] = \varkappa_+(t) \left( \sigma_+ \varrho \sigma_- -
\frac{1}{2} \{\varrho, \sigma_- \sigma_+ \} \right) +
\varkappa_-(t) \left( \sigma_- \varrho \sigma_+ - \frac{1}{2}
\{\varrho, \sigma_+ \sigma_- \} \right) + \varkappa_z(t) \left(
\sigma_z \varrho \sigma_z - \varrho \right)
\end{equation}

\noindent we get the following parameters $\lambda_s$,
$(\lambda_z)_s$, $(t_z)_s$ of $\Phi_s$:
\begin{equation} \label{lambdas-through-kappas}
\lambda_s = \frac{1}{s+\frac{1}{2}[(\varkappa_+)_s +
(\varkappa_-)_s + 4 (\varkappa_z)_s]} , \quad (\lambda_z)_s =
\frac{1}{s + (\varkappa_+)_s + (\varkappa_-)_s} , \quad (t_z)_s =
\frac{(\varkappa_+)_s - (\varkappa_-)_s}{s[s + (\varkappa_+)_s +
(\varkappa_-)_s]}.
\end{equation}

By Bernstein's theorem the Laplace transform $f_s$ of the
non-negative smooth function $f(t): [0,\infty) \rightarrow
[0,\infty)$ is completely monotone, i.e., $(-1)^n \frac{d^n
f_s}{ds^n} \geq 0$ for all $s \geq 0$ and $n=0,1,2,\ldots$. Hence,
we can characterize non-negative functions $f(t)$ in terms of
$f_s$. Applying this result to Eq.~\eqref{CP}, we conclude that
$\Phi(t)$ is completely positive if and only if both functions
$\frac{1}{s} - |\lambda_z|_s - |t_z|_s$ and $\frac{1}{s} + 2
(\lambda_z)_s + (\lambda_z^2)_s - 4 (\lambda^2)_s - (t_z^2)_s$ are
completely monotone. However, these conditions cannot be further
simplified in terms of the parameters $(\varkappa_+)_s$,
$(\varkappa_-)_s$, $(\varkappa_z)_s$ of the memory kernel. To
overcome this difficulty we resort to a subset of legitimate
quantum channels $\Phi$ that is given by linear restrictions on
parameters $\lambda,\lambda_z,t_z$.

\begin{lemma} \label{lemma}
Eq.~\eqref{Phi-channel} defines a completely positive and trace
preserving map $\Phi$ if
\begin{equation} \label{subset}
1 + 2\lambda + \lambda_z \pm t_z \geq 0, \quad 1 - 2\lambda +
\lambda_z \pm t_z \geq 0, \quad 1 - \lambda_z \pm t_z \geq 0.
\end{equation}
\end{lemma}
\begin{proof}
Conditions~\eqref{subset} define a convex polyhedron in the
parameter space $(\lambda,\lambda_z,t_z)$ with vertices (extremal
points) that satisfy restrictions~\eqref{CP}. Therefore,
conditions~\eqref{subset} define a convex hull of some quantum
channels.
\end{proof}

\begin{figure}
\includegraphics[width=8cm]{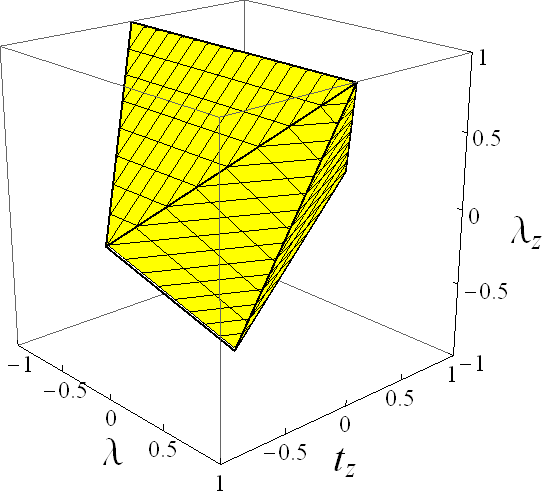}
\caption{\label{figure-5} A polyhedron in the parameter space
$(\lambda,\lambda_z,t_z)$ within which the map~\eqref{Phi-channel}
is completely positive by Lemma~\ref{lemma}.}
\end{figure}

Geometrically, conditions~\eqref{subset} define a body depicted in
Fig.~\ref{figure-5}. The edges connecting points
$(\lambda=1,\lambda_z=1,t_z=0)$ and
$(\lambda=0,\lambda_z=0,t_z=\pm 1)$ correspond to families of
shifted depolarizing channels. The edge connecting points
$(\lambda=1,\lambda_z=1,t_z=0)$ and
$(\lambda=0,\lambda_z=-1,t_z=0)$ comprises a trajectory of the
eternal CP-indivisible unital qubit process $\Phi_{a=0}(t) =
\Phi_{b=0}(t)$ from Section~\ref{section-eternal-indivisible}.

\begin{proposition} \label{prop-kernel}
The master equation $\frac{d \varrho(t)}{dt} = \int_0^t K(t')
[\varrho(t - t')] dt'$ with the memory kernel~\eqref{K-kernel}
defines a completely positive and trace preserving quantum
dynamics if six functions
\begin{equation} \label{kappa-conditions}
\frac{(\varkappa_{m})_s}{s[s+(\varkappa_+)_s + (\varkappa_-)_s]},
\quad \frac{s + (\varkappa_{m})_s}{s[s+(\varkappa_+)_s +
(\varkappa_-)_s]} \pm \frac{1}{s+\frac{1}{2}[(\varkappa_+)_s +
(\varkappa_-)_s + 4 (\varkappa_z)_s]}, \quad m=\pm,
\end{equation}

\noindent are all completely monotone.
\end{proposition}
\begin{proof}
$\Phi(t)$ is trace preserving by construction. $\Phi(t)$ is
completely positive if non-negativity conditions~\eqref{subset}
are satisfied. In the Laplace domain, this corresponds to complete
monotonicity of six functions $\frac{1}{s} + 2\lambda_s +
(\lambda_z)_s \pm (t_z)_s$, $\frac{1}{s} - 2\lambda_s +
(\lambda_z)_s \pm (t_z)_s$, $\frac{1}{s} - (\lambda_z)_s \pm
(t_z)_s$. Using the relation~\eqref{lambdas-through-kappas}, we
get functions~\eqref{kappa-conditions}.
\end{proof}

Following the idea of Ref.~\cite{chruscinski-2015}, we construct
an example illustrating Proposition~\ref{prop-kernel}.

\begin{example}
Let $a \geq a_{\pm} >0$ and $f(t)$ be a real-valued function such
that $0 \leq \int_0^{t} f(t') dt' \leq 2(a + a_{\pm})^{-1}$ for
all $t \geq 0$. The memory kernel~\eqref{K-kernel} with
coefficients
\begin{equation*}
(\varkappa_{\pm})_s = \frac{a_{\pm} s f_s}{1 - (a_+ + a_-)f_s},
\quad (\varkappa_z)_s = \frac{s f_s [2a - a_+ - a_- - a(a_+ +
a_-)f_s ]}{4(1 - a f_s)[1 - (a_+ + a_-)f_s]}
\end{equation*}

\noindent defines a legitimate master equation $\frac{d
\varrho(t)}{dt} = \int_0^t K(t') [\varrho(t - t')] dt'$.
Trajectories of this dynamical process in the parameter space
$(\lambda,\lambda_z,t_z)$ are given by segments of straight lines
because $\lambda_s = \frac{1}{s}(1-af_s)$, $(\lambda_z)_s =
\frac{1}{s}[1-(a_+ + a_-)f_s]$, $(t_z)_s = \frac{1}{s}(a_+ -
a_-)f_s$ and
\begin{equation*}
\lambda(t) = 1 - a \int_0^t f(t')dt', \quad \lambda_z(t) = 1 -
(a_+ + a_-) \int_0^t f(t')dt', \quad t_z(t) = (a_+ - a_-) \int_0^t
f(t')dt'.
\end{equation*}

\noindent Using the
relation~\eqref{gammas-through-lambdas-and-tz}, we conclude that
$\gamma_z(t) < 0$ and the dynamics is not CP-divisible if (i)
$f(t)<0$ or (ii) $f(t)
> 0$ and $a (a_+ + a_-) \int_0^t f(t')dt' > 2a - a_+ - a_-$ for
some $t>0$.
\end{example}

\section{Conclusions} \label{section-conclusions}

We have considered properties of time-local and convolution master
equations describing a phase covariant qubit dynamics. We have
characterized the parameters $\lambda$, $\lambda_z$, $t_z$ that
can be attained in semigroup dynamics with constant decoherence
rates (Proposition~\ref{prop-C-L} and Remark~\ref{remark-1}). We
have proved that this region of parameters cannot be significantly
extended if one allows for time-dependent non-negative rates
$\gamma_{\pm}(t)$ and $\gamma_z(t)$ in a time-local master
equation (Proposition~\ref{prop-C-CP} and Remarks~\ref{remark-2}
and \ref{remark-3}). We have clarified that the population can be
a non-monotonic function of time in a CP-divisible phase covariant
dynamics without regard to an initial system state
(Proposition~\ref{prop-non-monotonicity}). Then we have extended
the class of eternal CP-indivisible dynamics by presenting a
family of non-unital commutative dynamical maps
(Proposition~\ref{prop-eternal-commutative}) and a family of
non-unital non-commutative dynamical maps
(Proposition~\ref{prop-eternal-non-commutative}). The main results
of the paper are the positivity condition for a phase covariant
map (Proposition~\ref{prop-P}) obtained with the help of the
quantum Sinkhorn theorem and the condition for positive
divisibility (Proposition~\ref{prop-P-divisibility}). Finally, we
have considered a subset of completely positive phase covariant
qubit maps (Lemma~\ref{lemma}) with linear inequalities on
$\lambda$, $\lambda_z$, $t_z$ that we further used to specify a
class of admissible memory kernels in the convolution master
equation describing a phase covariant qubit dynamics.

The revealed divisibility properties have a close relation to
non-Markovianity of the system dynamics and an extended
system--ancilla dynamics, however, a discussion of this relation
is beyond the scope of this paper. Moreover, the same reduced
dynamics of the system $\Phi(t)$ may be caused by completely
different physical environments, and such a difference can be
revealed by interventions into the system dynamics, e.g., by
performing projective measurements on the system during the
evolution~\cite{pollock-2018,milz-2019,lvgf-2019}.

\begin{acknowledgements}
The authors thank Vladimir Frizen, David Davalos, Henri Lyyra,
Jose Teittinen, and Jyrki Piilo for fruitful discussions. The
study was supported by the Russian Science Foundation, project no.
19-11-00086.
\end{acknowledgements}


\begin{thebibliography}{10}

\bibitem{holevo-st}
 A.~S.~Holevo, \emph{Statistical Structure of
Quantum Theory} (Springer, Berlin, 2001).

\bibitem{davies-1976}
 E.~B.~Davies, \emph{Quantum Theory of Open Systems}
(Academic Press, London, 1976).

\bibitem{breuer-book}
 H.-P.~Breuer and F.~Petruccione, \emph{The Theory
of Open Quantum Systems} (Oxford University Press, Oxford, 2002).

\bibitem{holevo-book}
 A.~S.~Holevo, \emph{Quantum Systems, Channels,
Information. A Mathematical Introduction} (de Gruyter,
Berlin/Boston, 2012).

\bibitem{heinosaari-ziman}
 T.~Heinosaari and M.~Ziman, \emph{The Mathematical
Language of Quantum Theory} (Cambridge Univ. Press, Cambridge,
2012).

\bibitem{luchnikov-2019}
 I.~A.~Luchnikov, S.~V.~Vintskevich, H.~Ouerdane,
and S.~N.~Filippov, \textquotedblleft Simulation complexity of
open quantum dynamics: Connection with tensor
networks,\textquotedblright\, Phys. Rev. Lett. {\bf 122}, 160401
(2019).

\bibitem{alicki-book}
 R.~Alicki and K.~Lendi, \emph{Quantum Dynamical
Semigroups and Applications} (Springer, Berlin, 1987).

\bibitem{gks-1976}
 V.~Gorini, A.~Kossakowski, and
E.~C.~G.~Sudarshan, \textquotedblleft Completely positive
dynamical semigroups of $N$-level systems,\textquotedblright\, J.
Math. Phys. {\bf 17}, 821--825 (1976).

\bibitem{lindblad-1976}
 G.~Lindblad, \textquotedblleft On the generators
of quantum dynamical semigroups,\textquotedblright\, Commun. Math.
Phys. {\bf 48}, 119--130 (1976).

\bibitem{van-hove-1954}
 L.~van Hove, \textquotedblleft
Quantum-mechanical perturbations giving rise to a statistical
transport equation,\textquotedblright\, Physica {\bf 21}, 517--540
(1954).

\bibitem{davies-1974}
 E.~B.~Davies, \textquotedblleft Markovian master
equations,\textquotedblright\, Commun. Math. Phys. {\bf 39},
91--110 (1974).

\bibitem{palmer-1977}
 P.~F.~Palmer, \textquotedblleft The singular
coupling and weak coupling limits,\textquotedblright\, J. Math.
Phys. {\bf 18}, 527--529 (1977).

\bibitem{gorini-1978}
 V.~Gorini, A.~Frigerio, M.~Verri,
A.~Kossakowski, E.~C.~G.~Sudarshan, \textquotedblleft Properties
of quantum Markovian master equations,\textquotedblright\, Rep.
Math. Phys. {\bf 13}, 149--173 (1978).

\bibitem{accardi-book}
 L.~Accardi, Y.~G.~Lu, and I.~Volovich,
\emph{Quantum Theory and Its Stochastic Limit} (Springer, Berlin,
2002).

\bibitem{dumcke-1985}
 R.~D\"{u}mcke, \textquotedblleft The low density
limit for an $N$-level system interacting with a free Bose or
Fermi gas,\textquotedblright\, Commun. Math. Phys. {\bf 97},
331--359 (1985).

\bibitem{accardi-1991}
 L.~Accardi and Y.~G.~Lu, \textquotedblleft The
low-density limit of quantum systems,\textquotedblright\, J. Phys.
A: Math. Gen. {\bf 24}, 3483--3512 (1991).

\bibitem{apv-2002}
 L.~Accardi, A.~N.~Pechen, I.~V.~Volovich,
\textquotedblleft Quantum stochastic equation for the low density
limit,\textquotedblright\, J. Phys. A: Math. Gen. {\bf 35},
4889--4902 (2002).

\bibitem{hornberger-2007}
 K.~Hornberger, \textquotedblleft Monitoring
approach to open quantum dynamics using scattering
theory,\textquotedblright\, EPL {\bf 77}, 50007 (2007).

\bibitem{smirne-2010}
 A.~Smirne and B.~Vacchini, \textquotedblleft
Quantum master equation for collisional dynamics of massive
particles with internal degrees of freedom,\textquotedblright\,
Phys. Rev. A {\bf 82}, 042111 (2010).

\bibitem{rau-1963}
 J.~Rau, \textquotedblleft Relaxation phenomena
in spin and harmonic oscillator systems,\textquotedblright\, Phys.
Rev. {\bf 129}, 1880--1888 (1963).

\bibitem{scarani-2002}
 V.~Scarani, M.~Ziman, P.~\v{S}telmachovi\v{c},
N.~Gisin, and V.~Bu\v{z}ek, \textquotedblleft Thermalizing quantum
machines: Dissipation and entanglement,\textquotedblright\, Phys.
Rev. Lett. {\bf 88}, 097905 (2002).

\bibitem{luchnikov-2017}
 I.~A.~Luchnikov and S.~N.~Filippov,
\textquotedblleft Quantum evolution in the stroboscopic limit of
repeated measurements,\textquotedblright\, Phys. Rev. A {\bf 95},
022113 (2017).

\bibitem{filippov-2019}
 S. N. Filippov, G. N. Semin, A. N. Pechen,
\textquotedblleft Irreversible quantum dynamics for a system with
gas environment in the low-density limit and the semiclassical
collision model,\textquotedblright\, arXiv:1908.11202 [quant-ph].

\bibitem{exner-book}
 P. Exner, \emph{Open Quantum Systems and Feynman
Integrals}, corollary 2.4.10 (Reidel, Dordrecht, 1985).

\bibitem{holevo-1993}
 A.~S.~Holevo, \textquotedblleft A note on
covariant dynamical semigroups,\textquotedblright\, Rep. Math.
Phys. {\bf 32}, 211--216 (1993).

\bibitem{holevo-1996}
 A.~S.~Holevo, \textquotedblleft Covariant
quantum Markovian evolutions,\textquotedblright\, J. Math. Phys.
{\bf 37}, 1812--1832 (1996).

\bibitem{fk-2019}
 S.~N.~Filippov and K.~V.~Kuzhamuratova,
\textquotedblleft Quantum informational properties of the
Landau-Streater channel,\textquotedblright\, J. Math. Phys. {\bf
60}, 042202 (2019).

\bibitem{vacchini-2010}
 B.~Vacchini, \textquotedblleft Covariant
mappings for the description of measurement, dissipation and
decoherence in quantum mechanics,\textquotedblright\, Lecture
Notes in Physics {\bf 787}, 39--77 (2010).

\bibitem{ruskai-2002}
 M.~B.~Ruskai, S.~Szarek, and E.~Werner,
\textquotedblleft An analysis of completely-positive
trace-preserving maps on $M_2$, \textquotedblright\, Linear
Algebra Appl. {\bf 347}, 159--187 (2002).

\bibitem{smirne-2016}
 A.~Smirne, J.~Ko{\l}ody\'{n}ski, S.~F.~Huelga,
and R.~Demkowicz-Dobrza\'{n}ski, \textquotedblleft Ultimate
precision limits for noisy frequency
estimation,\textquotedblright\, Phys. Rev. Lett. {\bf 116}, 120801
(2016).

\bibitem{lankinen-2016}
 J.~Lankinen, H.~Lyyra, B.~Sokolov, J.~Teittinen,
B.~Ziaei, and S.~Maniscalco, \textquotedblleft Complete
positivity, finite-temperature effects, and additivity of noise
for time-local qubit dynamics,\textquotedblright\, Phys. Rev. A
{\bf 93}, 052103 (2016).

\bibitem{haase-2018}
 J.~F.~Haase, A.~Smirne, J.~Ko{\l}ody\'{n}ski,
R.~Demkowicz-Dobrza\'{n}ski, and S.~F.~Huelga, \textquotedblleft
Fundamental limits to frequency estimation: a comprehensive
microscopic perspective,\textquotedblright\, New J. Phys. {\bf
20}, 053009 (2018).

\bibitem{teittinen-2018}
 J.~Teittinen, H.~Lyyra, B.~Sokolov, and
S.~Maniscalco, \textquotedblleft Revealing memory effects in
phase-covariant quantum master equations,\textquotedblright\, New
J. Phys. {\bf 20}, 073012 (2018).

\bibitem{haase-2019}
 J.~F.~Haase, A.~Smirne, and S.~F.~Huelga,
\textquotedblleft Non-monotonic population and coherence evolution
in Markovian open-system dynamics,\textquotedblright\, in
\emph{Advances in Open Systems and Fundamental Tests of Quantum
Mechanics} edited by B.~Vacchini, H.-P.~Breuer, and A.~Bassi,
Springer Proceedings in Physics {\bf 237}, 41--57 (2019).

\bibitem{rivas-2014}
 \'{A}.~Rivas, S.~F.~Huelga, and M.~B.~Plenio,
\textquotedblleft Quantum non-Markovianity: characterization,
quantification and detection,\textquotedblright\, Rep. Prog. Phys.
{\bf 77}, 094001 (2014).

\bibitem{breuer-2016}
 H.-P.~Breuer, E.-M.~Laine, J.~Piilo, and
B.~Vacchini, \textquotedblleft Colloquium: Non-Markovian dynamics
in open quantum systems,\textquotedblright\, Rev. Mod. Phys. {\bf
88}, 021002 (2016).

\bibitem{de-vega-2017}
 I.~de~Vega and D.~Alonso, \textquotedblleft
Dynamics of non-Markovian open quantum
systems,\textquotedblright\, Rev. Mod. Phys. {\bf 89}, 015001
(2017).

\bibitem{benatti-2017}
 F.~Benatti, D.~Chru\'{s}ci\'{n}ski, and
S.~Filippov, \textquotedblleft Tensor power of dynamical maps and
positive versus completely positive
divisibility,\textquotedblright\, Phys. Rev. A {\bf 95}, 012112
(2017).

\bibitem{li-2018}
 L.~Li, M.~J.~W.~Hall, and H.~M.~Wiseman,
\textquotedblleft Concepts of quantum non-Markovianity: A
hierarchy,\textquotedblright\, Phys. Rep. {\bf 759}, 1 (2018).

\bibitem{pollock-2018}
 F.~A.~Pollock, C.~Rodr  guez-Rosario,
T.~Frauenheim, M.~Paternostro, and K.~Modi, \textquotedblleft
Operational Markov condition for quantum
processes,\textquotedblright\, Phys. Rev. Lett. {\bf 120}, 040405
(2018).

\bibitem{milz-2019}
 S.~Milz, M.~S.~Kim, F.~A.~Pollock, and K.~Modi,
\textquotedblleft Completely positive divisibility does not mean
Markovianity,\textquotedblright\, Phys. Rev. Lett. {\bf 123},
040401 (2019).

\bibitem{li-2019}
 C.-F.~Li, G.-C.~Guo, and J.~Piilo,
\textquotedblleft Non-Markovian quantum dynamics: What does it
mean?\textquotedblright\, EPL {\bf 127}, 50001 (2019).

\bibitem{ithier-2005}
 G.~Ithier, E.~Collin, P.~Joyez, P.~J.~Meeson,
D.~Vion, D.~Esteve, F.~Chiarello, A.~Shnirman, Y.~Makhlin,
J.~Schriefl, and G. Sch\"{o}n, \textquotedblleft Decoherence in a
superconducting quantum bit circuit,\textquotedblright\, Phys.
Rev. B {\bf 72}, 134519 (2005).

\bibitem{chernyavskiy}
 A.~Y.~Chernyavskiy, \textquotedblleft On the
fidelity of quantum gates under T1 and T2 relaxation,
\textquotedblright\, Proc. SPIE {\bf 11022}, 110222P (2019).

\bibitem{wolf-prl-2008}
 M.~M.~Wolf, J.~Eisert, T.~S.~Cubitt, and
J.~I.~Cirac, \textquotedblleft  Assessing non-Markovian quantum
dynamics,\textquotedblright\, Phys. Rev. Lett. {\bf 101}, 150402
(2008).

\bibitem{rivas-2010}
 \'{A}.~Rivas, S.~F.~Huelga, and M.~B.~Plenio,
\textquotedblleft Entanglement and non-Markovianity of quantum
evolutions,\textquotedblright\, Phys. Rev. Lett. {\bf 105}, 050403
(2010).

\bibitem{hall-2014}
 M.~J.~W.~Hall, J.~D.~Cresser, L.~Li, and
E.~Andersson, \textquotedblleft Canonical form of master equations
and characterization of non-Markovianity,\textquotedblright\,
Phys. Rev. A {\bf 89}, 042120 (2014).

\bibitem{megier-2016}
 N. Megier, D. Chru\'{s}ci\'{n}ski, J. Piilo, and
W. T. Strunz, \textquotedblleft Eternal non-Markovianity: from
random unitary to Markov chain realisations,\textquotedblright\,
Sci. Rep. {\bf 7}, 6379 (2017).

\bibitem{fpmz-2017}
 S.~N.~Filippov, J.~Piilo, S.~Maniscalco and
M.~Ziman, \textquotedblleft Divisibility of quantum dynamical maps
and collision models,\textquotedblright\, Phys. Rev. A {\bf 96},
032111 (2017).

\bibitem{bengtsson-2006}
 I.~Bengtsson and K.~\.{Z}yczkowski,
\emph{Geometry of Quantum States. An Introduction to Quantum
Entanglement} (Cambridge University Press, New York, 2006).

\bibitem{gurvits-2004}
 L.~Gurvits, \textquotedblleft Classical
complexity and quantum entanglement,\textquotedblright\, J.
Comput. Syst. Sci. {\bf 69}, 448--484 (2004).

\bibitem{aubrun-2015}
 G.~Aubrun and S.~J.~Szarek, \textquotedblleft
Two proofs of St{\o}rmer's theorem,\textquotedblright\,
arXiv:1512.03293 [math.FA] (2015).

\bibitem{aubrun-2017}
 G.~Aubrun and S.~J.~Szarek, \emph{Alice and Bob
Meet Banach: The Interface of Asympototic Geometry Analysis and
Quantum Information Theory}, section 2.4.3 (American Mathematical
Society, 2017).

\bibitem{fm-2017}
 S.~N.~Filippov and K.~Y.~Magadov,
\textquotedblleft Positive tensor products of maps and
$n$-tensor-stable positive qubit maps,\textquotedblright\, J.
Phys. A: Math. Theor. {\bf 50}, 055301 (2017).

\bibitem{ffk-2017}
 S.~N.~Filippov, V.~V.~Frizen and D.~V.~Kolobova,
\textquotedblleft Ultimate entanglement robustness of two-qubit
states against general local noises,\textquotedblright\, Phys.
Rev. A {\bf 97}, 012322 (2018).

\bibitem{filippov-romp-2018}
 S.~N.~Filippov, \textquotedblleft Lower and
upper bounds on nonunital qubit channel
capacities,\textquotedblright\, Rep. Math. Phys. {\bf 82},
149--159 (2018).

\bibitem{filippov-jms-2019}
 S.~N.~Filippov, \textquotedblleft Quantum
mappings and characterization of entangled quantum
states,\textquotedblright\, J. Math. Sci. {\bf 241}, 210--236
(2019).

\bibitem{blp}
 E.-M.~Laine, J.~Piilo, and H.-P.~Breuer,
\textquotedblleft Measure for the non-Markovianity of quantum
processes,\textquotedblright\, Phys. Rev. A {\bf 81}, 062115
(2010).

\bibitem{lorenzo-2013}
 S.~Lorenzo, F.~Plastina, and M.~Paternostro,
\textquotedblleft Geometrical characterization of
non-Markovianity,\textquotedblright\, Phys. Rev. A {\bf 88},
020102(R) (2013).

\bibitem{muller-hermes-2017}
 A.~M\"{u}ller-Hermes and D.~Reeb,
\textquotedblleft Monotonicity of the quantum relative entropy
under positive maps,\textquotedblright\, Ann. Henri Poincare {\bf
18}, 1777--1788 (2017).

\bibitem{nakajima-1958}
 S.~Nakajima, \textquotedblleft On quantum theory
of transport phenomena: Steady diffusion,\textquotedblright\,
Prog. Theor. Phys. {\bf 20}, 948--959 (1958).

\bibitem{zwanzig-1960}
 R.~Zwanzig, \textquotedblleft Ensemble method in
the theory of irreversibility,\textquotedblright\, J. Chem. Phys.
{\bf 33}, 1338--1341 (1960).

\bibitem{chruscinski-2019}
 D.~Chru\'{s}ci\'{n}ski, \textquotedblleft
Conditions for legitimate memory kernel master
equation,\textquotedblright\, in \emph{Classical and Quantum
Physics} edited by G.~Marmo, D.~Mart\'{\i}n de Diego,
M.~C.~Mu\~{n}oz Lecanda, Springer Proceedings in Physics {\bf
229}, 147--162 (2019).

\bibitem{chruscinski-2015}
 F.~A.~Wudarski, P.~Nale\.{z}yty, G.~Sarbicki,
and D.~Chru\'{s}ci\'{n}ski, \textquotedblleft Admissible memory
kernels for random unitary qubit evolution,\textquotedblright\,
Phys. Rev. A {\bf 91}, 042105 (2015).

\bibitem{chruscinski-2016}
 D.~Chru\'{s}ci\'{n}ski and A. Kossakowski,
\textquotedblleft Sufficient conditions for a memory-kernel master
equation,\textquotedblright\, Phys. Rev. A {\bf 94}, 020103(R)
(2016).

\bibitem{chruscinski-2017}
 D.~Chru\'{s}ci\'{n}ski and A. Kossakowski,
\textquotedblleft Generalized semi-Markov quantum
evolution,\textquotedblright\, Phys. Rev. A {\bf 95}, 042131
(2017).

\bibitem{chruscinski-2-2017}
 K. Siudzi\'{n}ska and D.~Chru\'{s}ci\'{n}ski,
\textquotedblleft Memory kernel approach to generalized Pauli
channels: Markovian, semi-Markov, and beyond,\textquotedblright\,
Phys. Rev. A {\bf 96}, 022129 (2017).

\bibitem{fc-2018}
 S.~N.~Filippov and D.~Chru\'{s}ci\'{n}ski,
\textquotedblleft Time deformations of master
equations,\textquotedblright\, Phys. Rev. A {\bf 98}, 022123
(2018).

\bibitem{davalos}
 D.~Davalos, M.~Ziman, and C.~Pineda,
\textquotedblleft Divisibility of qubit channels and dynamical
maps,\textquotedblright\, Quantum {\bf 3}, 144 (2019).

\bibitem{puchala}
 Z. Pucha{\l}a, {\L}. Rudnicki, and K.
\.{Z}yczkowski, \textquotedblleft Pauli semigroups and
unistochastic quantum channels, \textquotedblright\, Phys. Lett. A
{\bf 383}, 2376--2381 (2019).

\bibitem{nielsen}
 M. A. Nielsen and I. L. Chuang, \emph{Quantum
Computation and Quantum Information}, section 8.3.5 (Cambridge
University Press, Cambridge, 2000).

\bibitem{shaji-2005}
 A.~Shaji and E.~C.~G.~Sudarshan,
\textquotedblleft Who's afraid of not completely positive
maps?\textquotedblright\, Phys. Lett. A {\bf 341}, 48--54 (2005).

\bibitem{lvgf-2019}
 I.~A.~Luchnikov, S.~V.~Vintskevich,
D.~A.~Grigoriev, and S.~N.~Filippov, \textquotedblleft Machine
learning non-Markovian quantum dynamics,\textquotedblright\,
arXiv:1902.07019 [quant-ph] (2019).

\end{thebibliography}
\end{document}